\documentclass{article}

\usepackage{amsmath}
\usepackage{amsfonts}
\usepackage{amssymb}
\usepackage{amsthm}
\usepackage{enumerate}
\usepackage{fullpage}

\newtheorem{theorem}{Theorem}[section]
\newtheorem{lemma}[theorem]{Lemma}
\newtheorem{definition}[theorem]{Definition}

\newcommand{\F}{\mathbb{F}}

\newcommand{\mult}{\mathsf{mult}}
\newcommand{\fmult}{\mathsf{fmult}}

\usepackage{parskip}

\begin{document}
\title{On the CGGRT Criterion for Detecting Bipartite Perfect Matchings in NC}
\author{Swastik Kopparty\thanks{Department of Mathematics and Department of Computer Science, University of Toronto. Email: \texttt{swastik.kopparty@utoronto.ca}.   Research partially supported
by an NSERC Discovery Grant.} \and Shubhangi Saraf\thanks{Department of Mathematics and Department of Computer Science, University of Toronto, Canada. Email: \texttt{shubhangi.saraf@utoronto.ca}.   Research partially supported
by an NSERC Discovery Grant and a McLean Award.}}
\date{}
\maketitle
\thispagestyle{empty}

\begin{abstract}
The recent breakthrough work of Chatterjee, Ghosh, Gurjar, Raj and Thierauf~\cite{CGGRT26} gives the first deterministic NC algorithm for the bipartite matching problem. They show how to detect as well as find perfect matchings in bipartite graphs in NC. In this note we present an arguably simpler-to-state variation of the NC detection criterion of~\cite{CGGRT26}, with improved parameters. 


\end{abstract}
\newpage
\setcounter{page}{1}

\section{Introduction}
The problem of determining if a bipartite graph has a perfect matching is a central problem in complexity theory and algorithm design, and has been extremely well studied over several decades. Polynomial time algorithms for this problem have been known for a long time~\cite{Kuh55, FF56, HK73}.
The parallel-time complexity of this problem, however, remained open until very recently. Here we seek NC algorithms: computable in polylogarithmic parallel time with polynomial total number of operations. The first progress on this came from the algebraic algorithm of Lovasz~\cite{Lov79} (based on Edmonds~\cite{Edm67} and Tutte~\cite{Tut47}), which gave a {\em randomized} NC algorithm for detecting perfect matchings, via an NC rank computation algorithm~\cite{Csa76, IMR80, Ber84, Mulmuley87}. Ever since this work, the quest for obtaining deterministic NC algorithms for perfect matchings has been a major goal in derandomization~\cite{GPV88, FGT19}.


In their recent breakthrough work, Chatterjee, Ghosh, Gurjar, Raj and Thierauf~\cite{CGGRT26}, give the first deterministic NC algorithm for the bipartite matching problem. Not only does~\cite{CGGRT26} show how to decide in NC if a bipartite graph has a perfect matching, it also shows how to find the perfect matching in NC. We will only focus on the decision version. 

The \cite{CGGRT26} criterion for detecting perfect matchings is based on the notion of subspace designs, a pseudorandom object that was developed to analyze list decoding of algebraic codes~\cite{GX13, GK16}. More recently~\cite{BCDZ26} used subspace designs to give deterministic algorithms for some combinatorial problems about codes.

In this note, we present an arguably simpler-to-state variation of the criterion of~\cite{CGGRT26} for detecting perfect matchings, with improved parameters. The analysis of the criterion uses a strengthened form of the subspace design phenomenon for a slightly different setting (Lemma~\ref{lem:Wmult}), and closely mirrors the analysis by~\cite{CGGRT26}. 

\subsection{The Criterion}\label{sec:criterion}
Let $\F$ be a field with characteristic $> 3 n^3$. Let $\alpha_1,\ldots, \alpha_n, \beta_1, \ldots, \beta_n \in \F$ be distinct.

For $i \in [n]$, let $U_i \subseteq \F[X]$ be the $\F$-linear space:
$$ U_i =  \{ (X-\alpha_i)^k \cdot  A(X) \in \F[X] \mid \deg(A) \leq a \},$$
where $ k =  2n^3$, $ a = 2 n^2 - 2n +1$.

For $j \in [n]$, let $V_j \subseteq \F[X]$ be the $\F$-linear space:
$$ V_j = \{ (X-\beta_j)^{\ell} \cdot B(X) \in \F[X] \mid \deg(B) \leq b\},$$
where $ \ell = 2n^2$, $ b = 2n^3 - 2n + 1$.

\begin{theorem}\label{thm:main}
    Let $G$ be a bipartite graph with $n$ vertices in each part, and let $M \in \{0,1\}^{n \times n}$  be its bipartite adjacency matrix.

    Then $G$ does not have a perfect matching if and only if there exist vectors
    $$(P_1(X), \ldots, P_n(X)) \in U_1 \times U_2 \times \ldots \times U_n,$$
    $$(Q_1(X), \ldots, Q_n(X)) \in V_1 \times V_2 \times \ldots \times V_n,$$
    not both zero, such that:
    \begin{align}
    \label{eq:main}
        (P_1(X), \ldots, P_n(X)) \cdot M = (Q_1(X), \ldots, Q_n(X)).
    \end{align} 
\end{theorem}

By the works of~\cite{Csa76, IMR80, Ber84, Mulmuley87}, it follows that the criterion above can be checked in NC. More explicitly, one can show that the above criterion efficiently (in NC and with $\tilde O(n^6)$ operations) reduces to computing the rank of an $O(n^3) \times O(n^3)$ matrix over $\F$, improving on the $O(n^4) \times O(n^4)$ size matrix rank computation in the~\cite{CGGRT26} criterion from their Appendix A, and the $O(n^5) \times O(n^5)$ size matrix rank computation in the~\cite{CGGRT26} criterion from their Section 3.

The full details about the efficiency of the reduction to a rank computation appear in the appendix. Below we give a sketch of where the $O(n^3) \times O(n^3)$ dimensions of the matrix whose rank we are checking come from. 
Write $P_i(X)$ as $(X-\alpha_i)^k \cdot A_i(X)$.
Let $D$ be the $n \times n$ diagonal matrix whose $(i,i)$ entry equals $(X- \alpha_i)^k$.
Then there exist $(P_1(X), \ldots, P_n(X))$ and $(Q_1(X), \ldots, Q_n(X))$ 
as in Theorem~\ref{thm:main} satisfying \eqref{eq:main} if and only if
there exist $A_1(X), \ldots, A_n(X) \in \F[X]$, not all zero, such that for each $j \in [n]$,
the $j$'th entry of:
$$ (A_1(X), \ldots, A_n(X) ) \cdot D \cdot M$$
has its first $\ell$ derivatives at $\beta_j$ vanish.

There are $n \cdot (a+1) = O(n^3)$ variables specifying the coefficients of the $A_i(X)$, 
and there are $n \cdot \ell = O(n^3)$ linear equations which we are trying to satisfy. Thus checking the criterion of Theorem~\ref{thm:main} reduces to computing the rank of an $O(n^3) \times O(n^3)$ matrix with entries in $\F$.

\section{Proof of Theorem~\ref{thm:main}}

\subsection{Preliminaries: Subspaces and Multiplicities}

For a nonzero polynomial $f(X) \in \F[X]$, and $\alpha \in \F$, we define the {\em multiplicity of $f$ at $\alpha$}, denoted $\mult(f, \alpha)$, to be the highest power of $(X-\alpha)$ that divides $f(X)$.

We will use the Wronskian criterion for linear independence (see~\cite{BD10} for a history and quick proof).

For a collection of polynomials $f_1(X), \ldots, f_s(X) \in \F[X]$, we define the Wronskian polynomial $R(X)$ to be the determinant of the $s \times s$ matrix whose $(i,j)$ entry is
the $(j-1)$'st (formal) derivative $f_i^{(j-1)}(X)$.

\begin{lemma}\label{lem:W}
Suppose the degrees of $f_1(X), \ldots, f_s(X)$ are all less than the characteristic of $\F$.
Then $f_1(X), \ldots, f_s(X)$ are linearly independent over $\F$ if and only if their Wronskian polynomial $R(X)$ is nonzero.
\end{lemma}

The next lemma is an variation of the subspace design theorem of Guruswami-Kopparty~\cite{GK16}, which allows us to talk about the subspace design property for subspaces of varying codimension.

\begin{lemma}\label{lem:Wmult}
Let $W \subseteq\F[X]$ be an $s$-dimensional $\F$-linear space of polynomials with degree at most $d$, where $d$ is less than the characteristic of $\F$.
 
 For $\alpha \in \F$, define:
 $$ N_W(\alpha) = \max_{f \in W\setminus \{0\}} \mult(f, \alpha).$$

 Then:
 \begin{align}
  \label{eq:multbound}
  \sum_{\alpha \in \F } \max \bigg( N_W(\alpha) - s + 1 , 0 \bigg)  \leq  s d.
 \end{align}
\end{lemma}
Note: it turns out that $N_W(\alpha) - s + 1$ is always nonnegative (see Section~\ref{sec:generalized-sd}).
\begin{proof}
Let $R(X)$ be the determinant of the $s \times s$ Wronskian of a basis for $W$.

Then we have:
\begin{itemize}
\item $R(X)$ is nonzero, by Lemma~\ref{lem:W}.
 \item $\deg(R) \leq sd$.
 \item For any $\alpha\in \F$, $\mult(R, \alpha) \geq N_W(\alpha)-s + 1$. 
 To see this, let $f(X) \in W$ be such that $\mult(f, \alpha) = N_W(\alpha)$. By an $\F$-linear change of basis for $W$ (which only changes the Wronskian determinant by a nonzero $\F$-scalar), we can include $f$ in our basis for $W$. Expanding the determinant, and using the fact that the $j$'th derivative $f^{(j)}(X)$ is divisible by $(X-\alpha)^{N_W(\alpha) - j}$, 
 we get that one row of the Wronskian matrix, and hence the determinant, is divisible by $(X-\alpha)^{N_W(\alpha)-s+1}$, as desired.
\end{itemize}
This gives us the lemma.
\end{proof}

In the applications of Lemma~\ref{lem:Wmult}, $W$ will always be a linear space which is spanned by $s$ polynomials known to have high multiplicity at $s$ pre-chosen points. This will account for a large contribution to the LHS, and thus give a strong upper bound on $N_W(\alpha)$ for all  $\alpha$ not in the pre-chosen collection.

\subsection{Proof of Theorem~\ref{thm:main}}

We will index the vertices of both the left part and the right part of the graph $G$ by $[n]$. 
As a convention, the $n \times n$ bipartite adjacency matrix of $G$ has a $1$ in entry $(i,j)$
if and only if vertex $i$ on the left is adjacent to vertex $j$ on the right.

For a subset $I \subseteq [n]$ (viewed as left vertices),
we let $N(I) \subseteq [n]$ denote its neighbourhood (viewed as right vertices). The main tool for arguing about perfect matchings is Hall's theorem: $G$ has a perfect matching if and only if for all $I \subseteq [n]$, $|N(I)| \geq |I|$.

The first step in the proof of Theorem~\ref{thm:main} is a linear independence property of the $U_i$.

\begin{lemma}\label{lem:lin-indep}
 Let $I \subseteq [n]$.
 Suppose $(P_i )_{i \in I}$ are all nonzero with $P_i \in U_i$, then they are linearly independent over $\F$.
\end{lemma}
\begin{proof}
Let $I_0 \subseteq I$ be a minimal 
subset such that $(P_i)_{i \in I_0}$ spans
$(P_i)_{i \in I}$. 

Note that $(P_i)_{i \in I_0}$ are all linearly independent.
We will show that they do not span any nonzero element
from any $U_i$ with $i \not\in I_0$.

Let $W = \mathsf{span}( P_i(X) : i \in I_0 )$.
For $s = |I_0|$ and $d = k + a$, this satisfies the hypothesis of Lemma~\ref{lem:Wmult}.

Now observe that $N_W(\alpha_i) \geq k$ for each $i \in I_0$.
This already contributes $s \cdot (k - s + 1)$ to the LHS of 
~\eqref{eq:multbound}.
On the other hand, the RHS is only $sd = s \cdot (k + a)$.

Thus for any $\alpha \in \F \setminus \{ \alpha_i \mid i \in I_0 \}$, we have
$$N_W(\alpha) \leq  \left( s \cdot (k+a) - s \cdot (k-s+1) \right) + s - 1 = sa + s^2 - 1 \leq na  + n^2 -1 < k.$$

Thus $W$ does not intersect $U_i$ with $i \not\in I_0$, as desired.
\end{proof}

\paragraph{No perfect matching $\Rightarrow$ a solution to~\eqref{eq:main}:}

Suppose $G$ does not have a perfect matching.
We want to find a nontrivial solution to~\eqref{eq:main} with
$(P_1, \ldots, P_n) \in \prod U_i$ and $(Q_1, \ldots, Q_n) \in \prod V_j$.

By Hall's theorem, there exists $I \subseteq [n]$ 
such that the neighborhood $J = N(I)$ has $|J| < |I|$.
We will choose $(P_1, \ldots, P_n) \in \prod_{i \in [n]} U_i$  and $(Q_1, \ldots, Q_n) \in \prod_{j \in [n]}V_j$ satisfying $P_i = 0$ for $i \not\in I$ and $Q_j = 0$ for $j \not\in J$. Then only the minor $M_{I,J}$ of $M$ becomes relevant, and we seek a solution to:
$$(P_i)_{i \in I}  \cdot M_{I,J}  =  (Q_j)_{j \in J}.$$
We show that a nonzero solution exists by dimension counting.

The space of $(P_i)_{i \in I}$ with each $P_i \in U_i$ 
has dimension at least $a|I|$. The constraint that $(P_i)_{i \in I} \cdot M_{I,J} \in \prod_{j \in J} V_j$ imposes $\ell |J|$ constraints.

However, we have:
$$\frac{a |I|}{\ell |J|} \geq \frac{a}{\ell} \cdot \frac{n}{n-1} = \frac{2n^2 - 2n + 1}{2n^2} \cdot \frac{n}{n-1} > 1.$$
Thus there exists a nonzero solution $(P_i)_{i \in I}$, as desired.

\paragraph{A solution to~\eqref{eq:main} $\Rightarrow$ no perfect matching: }

Suppose there exists nonzero $(P_1, \ldots, P_n) \in \prod U_i$ and $(Q_1, \ldots, Q_n) \in \prod V_j$ 
with:
$$ (P_1, \ldots, P_n) \cdot M  = (Q_1, \ldots, Q_n).$$

Suppose further that $G$ has a perfect matching. 
We will derive a contradiction from this.

Let $I \subseteq [n]$ be the set of $i$ for which $P_i \neq 0$. Let $J \subseteq [n]$ be the neighborhood $N(I)$.
By Hall's theorem, $|J| \geq |I|$.

If we focus on the minor $M_{I,J}$ of $M$, we get that
$$(P_i)_{i \in I}  \cdot M_{I,J}  =  (Q_j)_{j \in J}.$$

Let
$$ W = \mathsf{span}( P_i  \mid i \in I ).$$
By Lemma~\ref{lem:lin-indep}, we have $\dim(W) = |I|$.
We will apply Lemma~\ref{lem:Wmult} to 
$W$, with $s = \dim(W) = |I|$ and $d = k + a$.

Observe that each $Q_j$ with $j \in J$ is a nonzero linear
combination of the $(P_i)_{i \in I}$, and thus by Lemma~\ref{lem:lin-indep}, lies in $W \setminus \{0\}$.
Since $Q_j$ also lies in $V_j$, then
we have that $N_W(\beta_j) \geq \ell$.

Then we can lower bound the LHS of equation~\eqref{eq:multbound},
summed at $\alpha \in \{ \alpha_i \mid i \in I\} \cup \{ \beta_j \mid j \in J \}$, by:
\begin{align*}
 LHS &\geq \sum_{i \in I} (N_W(\alpha_i) - s +1) + \sum_{j \in J}( N_W(\beta_j) - s + 1 )\\
 & \geq |I| \cdot (k - s + 1) + |J| \cdot (\ell - s + 1) \\
 & \geq s \cdot (k - s + 1) + s \cdot (\ell - s + 1) \\
 &\geq s \cdot (k + \ell - 2s + 2 ).
\end{align*}

On the other hand, we know that $LHS \leq sd = s(k+a)$, and so:
$$  s \cdot (k + \ell - 2s + 2 ) \leq s \cdot (k+a),$$
which tells us that 
$$\ell \leq a + 2s - 2 \leq a + 2n - 2 = \ell -1,$$
a contradiction.

\subsection{Relationship to the proof in~\cite{CGGRT26}}

The above presentation largely follows the proof of Theorem 3.1 in~\cite{CGGRT26},
with some small changes. We summarize the changes below for readers who are familiar with~\cite{CGGRT26}.
\begin{itemize}
    \item We instantiate everything with the Wronskian instead of the folded Wronskian. This is purely cosmetic, with the only advantage being that we can present things in terms of the well-known term ``multiplicity" instead of the less-well-known ``folded multiplicity":
    $$ \fmult(f, \alpha) = \mbox{the largest $i \geq 0$ s.t. } f(\alpha) = f(\alpha \gamma) = \ldots = f(\alpha \gamma^{i-1}) = 0.$$
    In the folded analogue of Theorem 1.1, we would choose $\gamma \in \F$ with multiplicative order $> 3n^3$,
    choose $\alpha_1, \ldots, \alpha_n, \beta_1, \ldots, \beta_n$ such that the ratio of every two of them is not in the set
    $\{1, \gamma, \ldots, \gamma^{3n^3} \}$, and define:
    $$U_i = \{ P(X) \in \F[X] \mid \deg(P) \leq k+a, \fmult(P, \alpha_i) \geq k \}.$$
    $$V_j = \{ Q(X) \in \F[X] \mid \deg(Q) \leq \ell + b, \fmult(Q, \beta_j) \geq \ell \}.$$
    
    \item In the matrix $\widehat{A}$ appearing in Theorem 3.1 in~\cite{CGGRT26}, there are $n-1$ copies of the identity matrix, which effectively creates high (folded) multiplicity vanishing constraints at $n-1$ $\F$-points per left vertex of the graph. In Theorem~\ref{thm:main}, we instead create a single $\F$-point, $\alpha_i$, per left vertex of the graph, and impose a very high multiplicity vanishing constraint at that point.

    To handle this, we have to prove a different subspace design theorem, Lemma~\ref{lem:Wmult}, to account for the varying multiplicities.

    \item The reduction in the number of $\F$-points where vanishing constraints are imposed (from $n \cdot (n-1)$ to $n$) has a quantitative advantage because of the way subspace design theorems accumulate some loss at each new point. The condition in Theorem~\ref{thm:main} amounts to computing the rank of an $O(n^3) \times O(n^3)$ matrix, which improves on the $O(n^5) \times O(n^5)$ matrix in Theorem 3.1 of~\cite{CGGRT26}, and the $O(n^4) \times O(n^4)$ matrix in the alternate algorithm in Appendix A of~\cite{CGGRT26}.

    Compared to Theorem 3.1 of~\cite{CGGRT26}, there are two factors of $n$ improved in the dimension of the matrix. One factor of $n$ is explained by the reduction in the number of $\F$-points. The other factor of $n$ is more benign, and is explained by the fact that elements of $U_i$ can be parametrized by polynomials of degree at most $a$, by multiplying them by $(X-\alpha_i)^k$; this is more efficient than taking all polynomials of degree $k+a$, and ensuring that they lie in $U_i$ by imposing $k$ derivative vanishing conditions at $\alpha_i$.
\end{itemize}

\section{A generalized subspace design theorem}\label{sec:generalized-sd}

In this section we give a generalized subspace design theorem for subspaces of varying codimension, generalizing both  Lemma~\ref{lem:Wmult} and the subspace design theorem from Guruswami-Kopparty~\cite{GK16}.
It is based on a finer study of the local behaviour of a subspace of polynomials around a point.
(Note that this section is not needed for the proof of Theorem~\ref{thm:main}.)

Let $W \subseteq \F[X]$ be an $\F$-linear space of dimension $s$.
Let $\alpha \in \F$.

We now define some parameters capturing the local behaviour of $W$ around $\alpha$,
analogous to successive minima in lattices. These can be best understood by writing
all polynomials in $W$ in the form $\sum_{j \geq 0} a_j (X-\alpha)^j$, and bringing the
vectors $(a_j)_{j \geq 0}$ into echelon form.

\begin{definition}
    For $1 \leq i \leq s$,
    define
    $$\lambda_i(W, \alpha) =  \max_{\substack{W' \subseteq W \\ \dim(W')=i} } \min_{f \in W' \setminus \{0\}} \mult(f, \alpha).$$
\end{definition}
Note that $\lambda_1(W,\alpha)$ is exactly the same as $N_W(\alpha)$ from Lemma~\ref{lem:Wmult}.

The following observations follow from the definition and elementary linear algebra:
\begin{itemize}
\item $\lambda_1(W,\alpha) > \lambda_2(W,\alpha) > \ldots > \lambda_s(W,\alpha) \geq 0$.
In particular, $\lambda_i(W,\alpha) \geq s - i$, and $N_W(\alpha) \geq s-1$.
    \item the space $W$ has a basis  $f_1, \ldots, f_s$ of the form:
$$ f_i(X) = (X- \alpha)^{\lambda_i(W, \alpha)} \cdot \hat{f}_i(X),$$
where $\hat{f}_i(\alpha) \neq 0$.
\end{itemize}

Taking the above basis for $W$, and using the fact that the $j$'th derivative $f_i^{(j)}(X)$ is divisible by $(X-\alpha)^{ \max(\lambda_i(W,\alpha)- j,\  0)}$, we get a natural lower bound on the vanishing multiplicity of the Wronskian $R(X)$ of $f_1, \ldots, f_s$ at $\alpha$:
$$ \mult(R, \alpha) \geq  \sum_{i=1}^s \max \bigg( \lambda_i(W,\alpha)- s + 1    , 0\bigg).$$

This gives us the following subspace design theorem, generalizing both the statement from~\cite{GK16} and Lemma~\ref{lem:Wmult}.

\begin{theorem}

 Let $W \subseteq\F[X]$ be an $s$-dimensional $\F$-linear space of polynomials with degree at most $d$, where $d$ is at most the characteristic of $\F$.

    Then:
    $$ \sum_{\alpha \in \F} \sum_{i=1}^s \max \bigg( \lambda_i(W,\alpha)- s + 1    , 0\bigg) \leq sd.$$
\end{theorem}

If we sum only over $\alpha, i$ where $\lambda_i(W, \alpha) \geq t$, we get the subspace design theorem from~\cite[Theorem 17]{GK16}.
If we sum only the $i=1$ terms, we get Lemma~\ref{lem:Wmult}.

\bibliographystyle{alpha}
\bibliography{refs.bib}

\appendix
\section{The $\F$-matrix underlying Theorem~\ref{thm:main}}

We now give some details of how the $O(n^3) \times O(n^3)$ matrix over $\F$, whose rank we are trying to compute, can be constructed  in NC and with $\tilde O(n^6)$ operations.

Write $P_i(X)$ as $A_i(X) \cdot (X-\alpha_i)^k$.
Represent $A_i(X)$ in terms of variables $A_{i,r}$ for $0 \leq r \leq a$:
$$A_i(X) = \sum_{r =0 }^{a} A_{i,r} X^r.$$

Let $D_t : \F[X] \to \F[X]$ denote the ($\F$-linear) operator of taking the $t$'th (Hasse) derivative: $ D_t( X^r) =  {r \choose t} X^{r-t}$.

Then the criterion of Theorem~\ref{thm:main} can be expressed as asking for the existence of $(A_{i,r})_{i \in [n], 0 \leq r \leq a}$, not all zero, such that for all $j \in [n], 0 \leq t < \ell$:
\begin{align*}
     D_t \left( \sum_{i : M_{ij} = 1 }   A_i(X)  \cdot (X-\alpha_i)^k  \right) (\beta_j)  = 0.
\end{align*} 
The LHS simplifies to:
\begin{align*}
&\sum_{i : M_{ij} = 1 }  D_t ( ( A_i(X)  \cdot (X-\alpha_i)^k ) ) (\beta_j)  \\
     &= \sum_{i : M_{ij} = 1 }   \sum_{t_1 + t_2 = t} D_{t_1}(A_i(X)) (\beta_j)  \cdot D_{t_2}((X-\alpha_i)^k ) ) (\beta_j) \\
     &= \sum_{i : M_{ij} = 1 }    \sum_{t_1 + t_2 = t} \left(\sum_{r \leq a}  A_{i,r}\cdot  {r \choose t_1} \cdot \beta_j^{r-t_1} \right) \cdot \left( {k \choose t_2} \cdot (\beta_j - \alpha_i)^{k-t_2}\right) \\
     &= \sum_{i : M_{ij} = 1 }  \sum_{r \leq a} A_{i,r}   \left( \sum_{t_1 + t_2 = t}  {r \choose t_1}\cdot  {k \choose t_2} \cdot \beta_j^{r-t_1} \cdot  (\beta_j - \alpha_i)^{k-t_2}\right) \\
     &= \sum_{i : M_{ij} = 1 }  \sum_{r \leq a} A_{i,r}  \cdot \Gamma_{i,r,j,t} \\
\end{align*} 
where $\Gamma_{i,r, j, t} \in \F$ is defined to be the following constant (independent of the input graph):
$$ \Gamma_{i,r,j,t} =   \sum_{t_1 + t_2 = t}  {r \choose t_1}\cdot  {k \choose t_2} \cdot \beta_j^{r-t_1} \cdot  (\beta_j - \alpha_i)^{k-t_2}.$$


Thus the criterion of Theorem~\ref{thm:main} reduces to computing the rank of a matrix,
with rows indexed by $[n] \times \{r : 0 \leq r \leq a\}$ and columns indexed by $[n] \times \{t : 0 \leq t < \ell\}$, whose $((i,r) , (j, t))$  entry equals:
$$ M_{ij} \cdot \Gamma_{i,r,j,t}.$$

Finally, notice that $\Gamma_{i,r,j,t}$ equals  $1/{ a \choose r }$ times the coefficient of $Y^r Z^t$ in:
$$ (1 + Y \cdot (Z + \beta_j) )^a \cdot ( Z + (\beta_j - \alpha_i) )^k.$$
Even though the $Z$-degree of this polynomial is $\Omega(n^3)$,  only monomials with $Z$ degree less than $\ell = O(n^2)$ are relevant.

Thus for fixed $i,j \in [n]$, all the relevant $\Gamma_{i,r,j,t}$ coefficients can be computed in NC with $\tilde{O}(n^4)$ operations by computing the above polynomial mod $Z^\ell$ using repeated squaring and the FFT based polynomial multiplication (which is in NC); see, for example, \cite{BCS97}. This can be done in parallel for all $i,j$, and thus the whole matrix can be computed in NC with $\tilde{O}(n^6)$ operations.
\end{document}